\documentclass[journal]{IEEEtran}

\ifCLASSINFOpdf
\else
\fi
\usepackage{cite,url,graphicx} 
\usepackage{amsmath}
\usepackage{amssymb}
\usepackage{amsfonts}
\usepackage{cite}
\usepackage{amsthm}
\newtheorem{theorem}{Theorem}
\usepackage{soul,xcolor}
\usepackage{pdfcomment}
\begin{document}
%
\title{Distributed Small-Signal Stability Conditions for  Inverter-Based Unbalanced Microgrids}
%
%
%
\author{Sai~Pushpak~Nandanoori,~\IEEEmembership{Member,~IEEE,}
        Soumya~Kundu,~\IEEEmembership{Member,~IEEE,}
        Wei~Du,~\IEEEmembership{Member,~IEEE,}
        Frank~K.~Tuffner,~\IEEEmembership{Member,~IEEE,}
        and~Kevin~P.~Schneider,~\IEEEmembership{Fellow,~IEEE}
\thanks{S. P. Nandanoori, S. Kundu, W. Du, F. K. Tuffner, and K. P. Schneider are with Pacific Northwest National Laboratory (PNNL), Richland, WA 99354 USA (e-mails:saipushpak.n@pnnl.gov, soumya.kundu@pnnl.gov, wei.du@pnnl.gov, francis.tuffner@pnnl.gov, kevin.schneider@pnnl.gov.}
\thanks{The PNNL is operated by Battelle for the U.S. Department of Energy under Contract DE-AC05-76RL01830.} 
}

\begin{table*}
\begin{large}
\textsuperscript{\textcopyright} 2020 IEEE.  Personal use of this material is permitted.  Permission from IEEE must be obtained for all other uses, in any current or future media, including reprinting/republishing this material for advertising or promotional purposes, creating new collective works, for resale or redistribution to servers or lists, or reuse of any copyrighted component of this work in other works.
\end{large}
\end{table*}

\newpage 
\maketitle

\begin{abstract}
The proliferation of inverter-based generation and advanced sensing, controls, and communication infrastructure have facilitated the accelerated deployment of microgrids. A coordinated network of microgrids can maintain reliable power delivery to critical facilities during extreme events. Low-inertia offered by the power-electronics–interfaced energy resources, however, can present significant challenges to ensuring stable operation of the microgrids. In this work, distributed small-signal stability conditions for inverter-based microgrids are developed that involve the droop-controller parameters and the network parameters (e.g. line impedances, loads). The distributed closed-form parametric stability conditions derived in this paper can be verified in a computationally efficient manner, facilitating the reliable design and operations of networks of microgrids. Dynamic phasor models have been used to capture the effects of electromagnetic transients. Numerical results are presented, along with PSCAD simulations, to validate the analytical stability conditions. Effects of design choices, such as the conductor types, and inverter sizes, on the small-signal stability of inverter-based microgrids are investigated to identify interpretable stable/unstable region estimates.
\end{abstract}


%
\IEEEpeerreviewmaketitle

\section*{Nomenclature}
\begin{tabular}{cl}
$\delta$	& Inverter terminal phase angle \\
$\omega$	& Inverter terminal frequency \\
$v$	& Inverter terminal voltage amplitude \\
$m_p$	& P-f droop coefficients \\
$m_q$	& Q-V droop coefficients \\
$\tau$	& Time constant of the low pass filter \\
$P_f$	& Filtered real power measurement \\
$Q_f$	& Filtered reactive power measurement \\
$Z^{abc}$	& Three phase impedance matrix of the network \\
$R^{abc}$	& Three phase resistance matrix of the network \\
$L^{abc}$	& Three phase inductance matrix of the network \\
$I_{ik}^{abc}$	& Three phase branch currents between buses $i,k$ \\ 
$I_{(l,i)}^{abc}$	& Three phase load currents at bus $i$ \\
$V_i^{abc}$	& Three  phase complex voltage at bus $i$ \\
$Y_{net}^{abc}$ &	Three phase network admittance matrix \\
$Y_{load}^{abc}$ & 	Three phase load admittance matrix \\
$Y_{net}(s)$	& Network admittance matrix in Laplace domain \\
$Y_{load}(s)$	& Load admittance matrix in Laplace domain \\
$\Lambda_p$	& Diagonal matrix whose entries are inverse \\
& of \textit{P-f} droop coefficients \\
\end{tabular}
\begin{tabular}{cl}
$\Lambda_q$	& Diagonal matrix whose entries are inverse \\ 
& of \textit{Q-V} droop coefficients \\
$\phi$	& Inverter phase angles (vector) \\	$\dot{\phi}$  & Inverter frequencies (vector) \\
$\rho$	& Inverter voltages (vector) \\
$T_1,T_2$	& Transformation matrices \\
$W, \dot{W}$ & 
Lyapunov function and its time derivative
\end{tabular}

\section{Introduction}
Microgrids are being deployed in increased numbers in recent years because they can be utilized as a resiliency resource to keep indispensable loads connected in case of an extreme event \cite{schneider2016evaluating}. With the proliferation of distributed energy resources, especially renewable energy resources (such as roof-top solar) and due to advancements in sensing, control, and communication technologies, real-time and autonomous coordination of distributed energy resources on a microgrid is increasingly possible \cite{schneider2019distributed,lasseter2011smart}. Two or more microgrids, when networked with each other, can offer significant benefits in operational efficiency, reliability and resilience during both normal conditions and extreme events; via effective utilization of local distributed generation units and flexible loads, \cite{lasseter2011smart,schneider2018enabling}.


Choice of various design parameters (e.g. droop gains, inverter size, conductor types) play a key role in determining the operational reliability and stability of inverter-based microgrids \cite{pogaku2007modeling,coelho2002small,du2014voltage}. A power sharing algorithm to load changes with sparse communication is presented in \cite{pushpak2015power}.  An explicit identification of parametric stability regions via simulation-based methods and eigenvalue analysis can be oftentimes computationally expensive, especially when the system size and/or the number of parameters involved is large, while not providing any robustness guarantees (e.g. stability margins). Many of the existing analytical methods used to obtain closed form estimates of the stability regions, on the other hand, suffer from various modeling assumptions that do not hold in practice (e.g. the works of \cite{simpson2013synchronization,schiffer2014conditions,vorobev2017framework,kundu2019identifying,kundu2019distributed,Vorobev2019decentralized}, as discussed in Sec.\,\ref{sec:review}). This inspired our work to consider realistic models and derive closed-form analytical stability conditions that can be verified in a computationally efficient manner.

\subsection{Review of Relevant Works}\label{sec:review}
Power from the alternative energy sources (e.g. renewable generation, energy storage units) are fed into the microgrid via inverters. Traditionally, inverters are designed to operate as a controlled current source regulating delivery of active (and reactive) power at its terminal, in what is known as the ``grid-following'' mode \cite{kroposki2017achieving}. Since grid-following inverters do not have the ability to regulate voltage or frequency, these cannot work in stand-alone mode. With advances in the inverter control technologies, voltage source converter-based ``grid-forming'' inverters are being increasingly adopted for their ability to regulate voltage and frequency in autonomous operation of islanded microgrids \cite{lasseter2011smart,chandorkar1993control,denis2018migrate}. This presents new challenges in modeling, analysis, and control of the transient behavior of the inverter-based microgrid, such as: 1) the shrinking gap in timescales of steady-state dispatch decisions and real-time control strategies, and 2) the inadequacy of the quasi-static phasor representations of the voltages and currents for transient analysis \cite{taylor2016power}. The choice of droop control gains in grid-forming inverters, 
along with network parameters and loading conditions, are critical to the stable operation of an inverter-based microgrid \cite{pogaku2007modeling}. In particular, droop gains associated with the inverters' outer power control loops have been identified as the defining factors for the dominant low-frequency eigenmodes, which determine the small-signal stability of the network \cite{coelho2002small,du2014voltage}.

In \cite{farrokhabadi2019microgrid}, the IEEE PES Task Force on microgrid stability discusses different notions of stability in microgrids and identifies the importance of small-signal stability, which they referred to as small-perturbation stability. In a review of stability issues in microgrids \cite{majumder2013some}, it was noted that the stability aspect depends on the type of microgrid (utility, facility, or remote), mode of operation (islanded or grid-connected), network parameters, and the control topology of the power electronic converters. In particular, the small-signal stability of microgrids has received much attention in the literature \cite{pogaku2007modeling,coelho2002small,du2014voltage,mendoza2014impedance,xu2016microgrids,yan2018small}, because ensuring operational reliability of a microgrid during small disturbances from steady state is considered necessary for successful design. In contrast to many related works in the literature, such as \cite{pogaku2007modeling,coelho2002small,du2014voltage,mendoza2014impedance,xu2016microgrids,yan2018small}, the focus of this present work is on deriving closed-form conditions of small-signal stability defined over various design parameters.

Identification of parametric stability conditions for inverter-based microgrids in terms of the droop gains has generated some interest recently \cite{simpson2013synchronization,schiffer2014conditions,vorobev2017framework,kundu2019identifying,kundu2019distributed,Vorobev2019decentralized}. In \cite{simpson2013synchronization}, the authors used a lossless microgrid model to draw the similarity between droop-controlled inverters and Kuramoto oscillators, and prescribed droop gain values for frequency synchronization and desirable power sharing. Port-Hamiltonian representation of the lossless microgrid was adopted in \cite{schiffer2014conditions} to derive parametric stability conditions on droop gains and power setpoints. Interestingly, stability conditions were found to be independent of frequency-droop controller gains and set points but depend on the voltage-droop controller gains and setpoints. A main drawback in the above two works is the assumption of a lossless microgrid, which is invalid in low-to-medium voltage systems with non-negligible line resistance-to-reactance ratios. This assumption has been relaxed in \cite{vorobev2017framework,kundu2019identifying,kundu2019distributed,Vorobev2019decentralized}, where the authors derived stability conditions on the droop coefficients for a (lossy) microgrid, which are verifiable locally at each inverter node. In \cite{vorobev2017framework}, the authors derive closed-form distributed stability conditions involving droop gains and line parameters. However, the conditions were derived for a specific microgrid model with no shunt elements (e.g., loads, shunt capacitors). A sum-of-squares optimization method was used in \cite{kundu2019identifying} to compute robust distributed stability certificates on the inverter droop gains, under bounded and time-varying uncertainties. A computational method using the dissipativity approach was proposed in \cite{Vorobev2019decentralized} to estimate stability certificates in different frequency regions. The aforementioned works, with the exception of \cite{vorobev2017framework}, use the conventional quasi-static phasor models, while a dynamic phasor-based modeling approach (\cite{sanders1991generalized,guo2014dynamic,vorobev2017high}) is often better suited to capture the electromagnetic transient effects in the stability of microgrids. 
Finally, all of the above works assume a balanced system, while low-to-medium voltage microgrids are typically unbalanced. In summary, it can be observed that several key assumptions are made in the above-mentioned works, including lossless network (in \cite{simpson2013synchronization} and \cite{schiffer2014conditions}), absence of shunt elements (in \cite{vorobev2017framework}), quasi-static phasor models (all, except \cite{vorobev2017framework}), and balanced system (all), which limit their applicability to realistic microgrids in which those assumptions do not often hold. 

\subsection{Summary of Contributions}
In this work, a set of closed-form, distributed, parametric small-signal stability conditions, involving droop-controller parameters (e.g. droop gains) and network parameters (e.g. line impedances), are derived for inverter-based microgrids. The main contributions of this work lie in: 1) consideration of realistic network models, including three-phase unbalanced, lossy microgrids with shunt elements (loads, capacitors); 2) adoption of dynamic phasor-based approach to represent the electromagnetic transient effects in a microgrid; and 3) validation of the analytical stability conditions with PSCAD-based simulations, along with a study of the conservativeness of the analytical stability conditions for various conductor types, and inverter sizes and locations.

The rest of the article is structured as follows: Section \ref{sec:microgrid_model}  describes the microgrid modeling approach using dynamic phasors; Section \ref{sec:main_result}   develops the distributed stability certificates using LaSalle’s invariance principle; Section \ref{sec:results}  presents numerical results to validate the approach; while conclusions and future directions are presented in Section \ref{sec:conclusion}.  

\section{Description of Microgrid Model}
\label{sec:microgrid_model}
This section describes the full-scale electromagnetic model for microgrids with multi-loop droop-controlled inverters, and then presents a dynamic phasor-based reduced-order model.
\subsection{Full-Scale Electromagnetic Model}
In this paper, a network of $N$ grid-forming, droop-controlled inverters operated in islanded mode \cite{lasseter2011smart,denis2018migrate}, with constant-impedance loads is considered. Constant-power and constant-current loads can also be modeled in the form of their equivalent impedances when the variations in voltages are small \cite{vorobev2017high,bottrell2013dynamic}.

A grid-forming inverter is designed to operate as a voltage source regulating both the voltage magnitude and the frequency at the point of coupling to the network \cite{kroposki2017achieving}. The inverters are implemented with the commonly used multi-loop droop-control system (shown in Fig. \ref{fig:multi_loop_droop}) \cite{pogaku2007modeling,li2004design,guerrero2011hierarchical}, comprising the inner current loop and outer voltage loop \cite{vorobev2017framework}. The inner loop is designed to be faster than the outer one, which allows independent tuning of the inner and outer control loops. 
\begin{figure}[h!]
    \centering
    \vspace{-0.3 cm}
    \includegraphics[width = 0.85 \columnwidth]{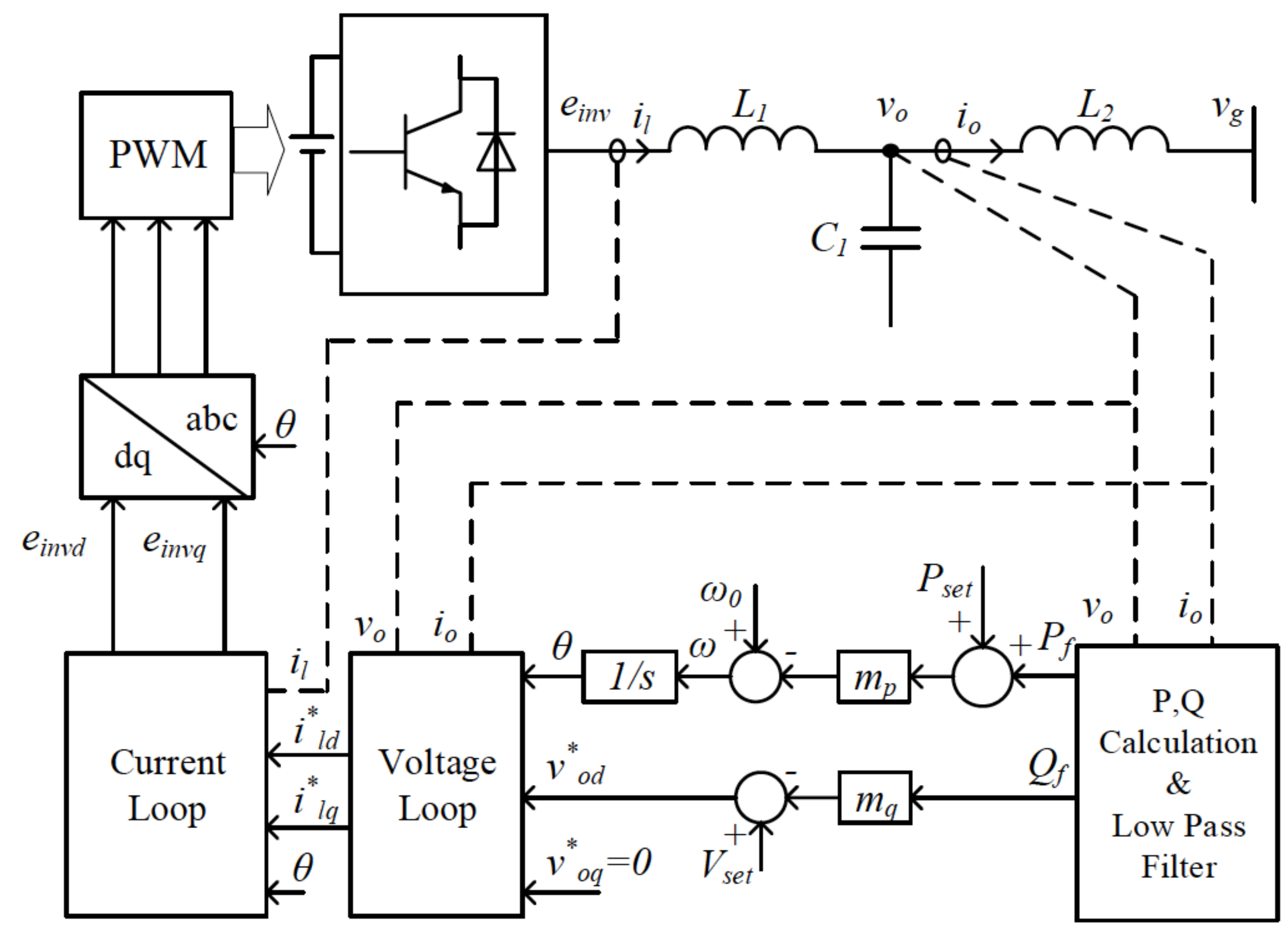}
    \caption{The controller block diagram of a multi-loop droop-controlled inverter.}
    \vspace{-0.3 cm}
    \label{fig:multi_loop_droop}
\end{figure}

The controller has a cascaded structure, including the power vs frequency (\textit{P-f}) and reactive power vs. voltage (\textit{Q-V}) droop-control loops, the voltage-control loop, and the current-control loop. This control strategy is designed to regulate the angular frequency and magnitude of the inverter filter capacitor voltage, $v_o\angle\delta$, according to the \textit{P-f} droop and \textit{Q-V} droop, respectively, such that the inverter may appear as an (almost) ideal voltage source at its terminal. In particular, the \textit{P-f} droop control ensures that the angular frequency $\omega$ of the capacitor voltage decreases whenever the output power of the inverter increases, or vice versa, due to some disturbance/changes in the network.

In contrast, due to an external disturbance, the \textit{Q-V} droop control ensures that the magnitude of the capacitor voltage decreases whenever the output reactive power increases, or vice versa. In Fig. \ref{fig:multi_loop_droop}, $m_p$  and $m_q$  are the \textit{P-f} and \textit{Q-V} droop coefficients, respectively; $P_f$ and $Q_f$ are, respectively, the filtered values of active and reactive power measurement (summed over all phases); $P_{set}$  and $Q_{set}$ are the active and reactive power set points, respectively; $\omega_o$ is the rated angular frequency; $V_{set}$  is the voltage magnitude set point, $e_{inv}$ is the inverter internal voltage, $v_{od}^*$ and $v_{oq}^*$ are the input references of the voltage loop, while $i_{ld}^*$ and $i_{lq}^*$ are the input references of the current loop. The reactive power set point, $Q_{set}$ is assumed to be zero and hence not seen in Fig. \ref{fig:multi_loop_droop}. In this work, it is assumed that the parameters of inner voltage and current loops are well designed to achieve fast control of the filter capacitor voltage. So the dynamics of inner controllers are ignored in this paper. Investigating how the parameters of inner controller affect the stability is out of the scope of this paper. The work in \cite{du2019comparative} investigated how inappropriate parameters of inner controllers affect the microgrid stability. Therefore, bypassing the faster internal states of the droop-controlled inverter, the system can be described using only the inverter terminal states (phase angle, frequency, and voltage amplitude) and the line currents as dynamic variables, as shown in  \cite{vorobev2017framework,vorobev2017high}. The dynamics of an inverter at bus $i$ can be written as follows:
\begin{equation}
    \begin{aligned}
    \dot{\delta}_i = & \omega_i - \omega_o \\
    \tau \dot{\omega}_i = & \omega_o - \omega_i + m_{p,i}\left(P_{set,i} - P_{f,i}\right) \\
    \tau \dot{v}_{o,i} = & v_{set} - x_{o,i} + m_{q,i} \left(Q_{set,i}-Q_{f,i}\right)
    \end{aligned}
    \label{eq:inverter_dynamics}
\end{equation}
where $\tau$ denotes the time constant of the measurement filter (low pass filter in Fig. \ref{fig:multi_loop_droop}). Note that the control algorithm ensures that the voltage at the inverter terminal appears as balanced, with the magnitude of each phase voltage being $v_{o,i}$ and the phases shifting uniformly by amount $\delta_i$. As discussed in \cite{vorobev2017high} and \cite{guo2014dynamic}, it is convenient to use a dynamic phasor modeling approach for small-signal stability analysis of inverter-based microgrids. Dynamic phasors are slow-varying Fourier coefficients describing the time-varying signal of interest \cite{sanders1991generalized}. Applying the dynamic phasor analysis around the nominal frequency (60 Hz), the three-phase line current dynamics corresponding to every pair of neighboring buses $(i,k)$ can be written as follows:
\begin{equation}
    L_{ik}^{abc} \dot{I}_{ik}^{abc} = V_i^{abc}- V_k^{abc} - Z_{ik}^{abc} I_{ik}^{abc}
    \label{eq:line_dynamics}
\end{equation}
where $Z_{ik}^{abc} = R_{ik}^{abc} + j\omega_o L_{ik}^{abc}$ is the $3\times 3$ matrix representing the line impedance; $R_{ik}^{abc}$ and $L_{ik}^{abc}$ are the line resistance and inductance matrices, respectively; $I_{ik}^{abc}$ is a (complex) vector that denotes the dynamic phasor of the line current across three phases; while $V_i^{abc}$ and $V_k^{abc}$ are three-dimensional (complex) vectors denoting the dynamic phasors of the voltages at the corresponding buses. The impedance matrices $R_{ik}^{abc}$ and $L_{ik}^{abc}$ are symmetric. Moreover, $R_{ik}^{abc}=R_{ki}^{abc}$ and $L_{ik}^{abc}=L_{ki}^{abc}$ for every $(i,k)$. Similar equations describe the three-phase unbalanced load current dynamics for every load attached to a phase $\alpha$ at bus $i$.
\begin{equation}
    L_{l,i}^{abc} \dot{I}_{l,i}^{abc} = V_i^{abc}- Z_{l,i}^{abc} I_{l,i}^{abc} 
    \label{eq:load_dynamics}
\end{equation}
where $Z_{l,i}^{abc}=R_{l,i}^{abc} + j \omega_o L_{l,i}^{abc}$ is the $3\times3$ matrix representing the (equivalent) impedance of the load attached at bus $i$; $R_{l,i}^{abc}$ and $L_{l,i}^{abc}$ are the (equivalent) load resistance and inductance, respectively; and $I_{l,i}^{abc}$ denotes the (complex) dynamic phasor of the load current across three phases.

Finally, the total active and reactive power flowing out of the inverter terminal across three phases at bus $i$ is given by
\begin{equation}
    \begin{aligned}
    P_{f,i} = & \mbox{Re}\left[\left(V_i^{abc}\right)^{\top} \cdot \mbox{conj}\left(I_i^{abc}\right)\right]  \\
    Q_{f,i} = & \mbox{Im}\left[\left(V_i^{abc}\right)^{\top} \cdot \mbox{conj}\left(I_i^{abc}\right)\right]  \\
    \end{aligned}
    \label{eq:power_equations}
\end{equation}
where $I_i^{abc} = I_{l,i}^{abc} + \sum_k I_{ik}^{abc}$ with $I_i^{abc}$ being the three-dimensional vector of the total current injected at three phases of bus $i$. Equations \eqref{eq:inverter_dynamics}-\eqref{eq:power_equations} describe the dynamical model of the microgrid network, resulting in a high-dimensional system. Singular perturbation theory was applied to further reduce the dimension of the microgrid model for analytical studies \cite{vorobev2017high,rasheduzzaman2015reduced,mariani2014model}.
\subsection{Reduced-Order Dynamic Model}
Electromagnetic transients that govern the line currents occur on the millisecond time scale, prompting a model reduction approach using singular perturbation theory in which the left-hand side of the line dynamics \eqref{eq:line_dynamics} is set to zero \cite{rasheduzzaman2015reduced,mariani2014model}. A model reduction approach such as this can be referred to as the zeroth-order simplification. The argument is that the inverter dynamics (governed by the measurement filter time constant $\tau_i$) is sufficiently slow that, for small-signal stability analysis, the line currents can be assumed to be in quasi-steady state. It was shown in  \cite{vorobev2017framework} and \cite{vorobev2017high} however, that zeroth-order simplification may lead to incorrect stability assessments for inverter-based microgrids. Instead, the authors in \cite{vorobev2017high} proposed a first-order simplification method allowing for improved accuracy in inclusion of the faster time scales in the slower dynamic modes. In particular, modes that are much slower than the electromagnetic transients are of interest. Transforming \eqref{eq:line_dynamics} and \eqref{eq:load_dynamics} into the Laplace domain, and performing first-order Taylor series approximation on the Laplace variable $s$ of the transfer functions from voltages to currents, the resultant equations are 
\begin{equation*}
    \begin{aligned}
    I_{ik}^{abc}(s) \approx & \left(Y_{ik}^{abc} - s Y_{ik}^{abc} L_{ik}^{abc} Y_{ik}^{abc} \right) \left(V_i^{abc}(s) - V_k^{abc}(s) \right) \\
    I_{l,i}^{abc}(s) \approx & \left(Y_{l,i}^{abc} - sY_{l,i}^{abc} L_{l,i}^{abc} Y_{l,i}^{abc}\right) V_i^{abc}(s)
    \end{aligned}
\end{equation*}
where the symmetric matrices, $Y_{ik}^{abc} = \left(Z_{ik}^{abc}\right)^{-1}$ and $Y_{l,i}^{abc} = \left(Z_{l,i}^{abc}\right)^{-1}$ represent the (equivalent) admittances of the line and the load, respectively. Setting $s=0$, the equations yield the quasi-steady-state relation between the voltages and currents. The relation between the voltages and injected currents at the buses across the network can be compactly written in the Laplace domain as follows. 
\begin{equation}
    \begin{bmatrix}
    I_1^{abc}(s) \\ I_2^{abc}(s) \\ \vdots \\ I_N^{abc}(s)  
    \end{bmatrix} \approx \left(Y_{net}^{abc}(s) + Y_{load}^{abc}(s)\right) \begin{bmatrix}
    V_1^{abc}(s) \\ V_2^{abc}(s) \\ \vdots \\ V_N^{abc}(s)  
    \end{bmatrix}
    \label{eq:Laplace_equations_compact}
\end{equation}
where $Y_{net}^{abc}(s)$ and $Y_{load}^{abc}(s)$ are symmetric, complex, $3N\times3N$ admittance matrices affine in the Laplace variable $s$, referred to as the network admittance matrix and the load admittance matrix, respectively. $Y_{load}^{abc}(s)$ is a diagonal matrix, with the entries of the $3\times3$ diagonal block that correspond to bus $i$ being $Y_{l,i}^{abc} - s Y_{l,i}^{abc} L_{l,i}^{abc} Y_{l,i}^{abc}$ whenever there is a load attached to it. $Y_{net}^{abc}(s)$ is a symmetric matrix with the entries of the $3\times3$ off-diagonal blocks corresponding to any pair of neighboring buses $(i,k)$ being $s Y_{ik}^{abc} L_{ik}^{abc} Y_{ik}^{abc} - Y_{ik}^{abc}$, while the diagonal block corresponding to bus $i$ is the negative of the sum of the off-diagonal blocks for every neighboring pair $(i,k)$. Recall that the inverter terminal voltages are balanced, i.e., the three-phase voltages have the same magnitude and are $120^{o}$ apart from each other. In other words, $V_i^b(s) = t V_i^a(s)$ and $V_i^c(s) = t^2 V_i^a(s)$, with $t = e^{-j2\pi/3}$. Therefore, 
\begin{equation}
    \begin{bmatrix}
    V_{o,1}^{abc}(s) \\ V_{o,2}^{abc}(s) \\ \vdots \\ V_{o,N}^{abc}(s) 
    \end{bmatrix}
    = \Theta \begin{bmatrix}
    V_{o,1}^{a}(s) \\ V_{o,2}^{a}(s) \\ \vdots \\ V_{o,N}^{a}(s)\end{bmatrix}, \quad \Theta = \pmb{I}_N \otimes \begin{bmatrix} 1 \\ t \\ t^2 \end{bmatrix}
    \label{eq:3_1_trans}
\end{equation}
where $\otimes$ represents a Kronecker product and $\pmb{I}_N$ is an $N\times N$ identity matrix. Let us define the following $N\times N$ matrices. 
\begin{equation}
    \begin{aligned}
    Y_{net}(s) := & \Theta^H \cdot Y_{net}^{abc}(s) \cdot \Theta = Y_{net}^0 + s Y_{net}^1 \\
    \mbox{and} Y_{load}(s) := & \Theta^H \cdot Y_{load}^{abc}(s) \cdot \Theta = Y_{load}^0 + s Y_{load}^1 \\
    \end{aligned}
    \label{eq:Y_transf}
\end{equation}
where $Y_{net}^0,Y_{net}^1,Y_{load}^0,Y_{load}^1$ are complex $N\times N$ matrices, and $(\cdot)^H$ denotes the Hermitian of a matrix. 

Linearizing the system of equations around an operating point of interest that corresponds to small phase angle differences between the buses and (close to) 1 p.u. bus voltage magnitudes, the microgrid model can be expressed by the following set of differential equations (see \cite{vorobev2017high} for details).
\begin{equation}
    \begin{aligned}
    \tau {\Lambda}_p \Ddot{\phi} + ({\Lambda}_p - {B}^{\prime}) \dot{\phi} + {B}\phi + ({G}+\tilde{{G}} )\rho - {G}^{\prime} \dot{\rho} & = 0 \\
    (\tau {\Lambda}_q - {B}^{\prime}) \dot{\rho} + ({\Lambda}_q + {B} + \tilde{{B}}) \rho - {G} \phi + {G}^{\prime} \dot{\phi} & = 0
    \end{aligned}
    \label{eq:linear_microgrid}
\end{equation}
where $\phi$ and $\rho$ are the $N$-dimensional vectors of deviations of the bus voltage angles and magnitudes, respectively, from their nominal operating point values; ${\Lambda}_p$ and ${\Lambda}_q$ are $N \times N$ diagonal matrices with diagonal entries equal to the inverse of the \textit{P-f} and \textit{Q-V} droop coefficients, respectively; and the $N \times N$ complex matrices $B,G,\tilde{B}, \tilde{G}, {B}^{\prime}$ and ${G}^{\prime}$ are given by
\begin{equation}
  \begin{aligned}
  B = -\mbox{Im}\left[Y_{net}^0\right], &\; G = \mbox{Re}\left[Y_{net}^0\right]\\
  \tilde{B} = -2\mbox{Im}\left[Y_{load}^0\right], &\; \tilde{G} = 2\mbox{Re}\left[Y_{load}^0\right]\\
  B^{\prime} = \mbox{Im}\left[Y_{net}^1 + Y_{load}^1\right], &\; G^{\prime} =  -\mbox{Re}\left[Y_{net}^1 + Y_{load}^1\right]
  \end{aligned}
  \label{eq:matrices_definition}
\end{equation}
As per definition, (1) $B$ and $G$ are positive semidefinite matrices (denoted as $B \succeq 0, G \succeq 0$); (2) $\tilde{B}$ and $\tilde{G}$ are diagonal, positive-definite matrices (denoted by $\tilde{B}\succ 0, \tilde{G} \succ 0$); while (3) $B^{\prime}$ is a positive-definite matrix ($B^{\prime} \succ 0$) and $G^{\prime}$ is a symmetric (but sign-indefinite) matrix. Finally, with the assumption that $\Lambda_q-B^{\prime}/\tau$ to be nonsingular, the linearized system \eqref{eq:linear_microgrid} has an unique equilibrium point at origin. Equations \eqref{eq:linear_microgrid} represent the reduced-order dynamical model for the unbalanced microgrid network, capable of accurately describing the dynamic modes that evolve on time scales slower than the electromagnetic time scale (see \cite{vorobev2017high} for detailed comparison). In this work, the reduced model of \eqref{eq:linear_microgrid} will be used for the analytical calculations of stability conditions, while the results will be validated with full-scale electromagnetic models. 

\section{Main Result: Parametric Stability Conditions}
\label{sec:main_result}
Note that the system of equations in \eqref{eq:linear_microgrid} can be written in the form of a linear, time-invariant dynamical system. Small-signal stability of the microgrid network can be guaranteed by computing the eigenvalues of the system matrix and ensuring that those lie on the left half-plane. However, such a method of determining stability via explicit enumeration of eigenvalues is computationally inefficient if the purpose is to identify a range of system parameter values for which the network is small-signal stable; the computational time grows exponentially with the increase in the number of parameters of interest (such as network impedances, loads, generation set points, droop coefficient values, etc.). Moreover, rather than an explicit computation of the parametric stability regions, a closed-form estimation of an inner approximation of the stability region is often more desirable. This is due to the facts that (1) closed-form expressions can be used as constraints in an economic optimization problem where design decisions are being made to minimize investment and operational costs; and (2) inner approximation of the stability region provides a natural safeguard against modeling inaccuracies. Lyapunov-based analysis provides a tractable alternative to eigenvalues analysis in estimating the parametric stability region.
\subsection{Lyapunov Reformulation of Stability Conditions}
Consider the system of \eqref{eq:linear_microgrid}, which can be written in state-space form with the help of the $3N$-dimensional states vector $x = \begin{bmatrix} \phi^{\top},  \rho^{\top},  \dot{\phi}^{\top} \end{bmatrix}^{\top}$. Applying LaSalle’s invariance principle \cite{slotine1991applied}, it can be argued that if there exist $3N\times 3N$ matrices $\tilde{\Psi} \succ \pmb{0}$ (positive definite) and $\tilde{\Pi} \succeq \pmb{0}$ (positive semidefinite) satisfying 
\begin{equation}
    W(x) ;= x^{\top} \tilde{\Psi} x \implies \dot{W}(x) \leq -x^{\top} \tilde{\Pi} x \leq 0, 
    \label{eq:Lyapunov_func_derivative}
\end{equation}
then the system trajectories are guaranteed to asymptotically converge to the largest invariant set contained in the set $\dot{W}(x) = 0$ which is the origin itself as the linearized system \eqref{eq:linear_microgrid} has a unique equilibrium at origin. That is to say the microgrid is small-signal stable. $W(x)$ is a Lyapunov function that is radially unbounded and positive definite. Finding such a pair of matrices $\tilde{\Psi}$ and $\tilde{\Pi}$ is generally difficult. Note, however, that if there exists a positive-definite matrix $\Psi (\Psi \succ \pmb{0})$ and a positive semidefinite matrix $\Pi (\Pi \succeq \pmb{0})$, as well as transformation matrices, $T_1$ and $T_2$, with $T_1$ being full rank, such that \eqref{eq:Lyap_derivative_transformed} holds, 
\begin{equation}
    W = y^{\top} \Psi y \implies \dot{W} \leq -z^{\top} \Pi z \leq 0
    \label{eq:Lyap_derivative_transformed}
\end{equation}
where $y = T_1 x$ and $z=T_2 x$, then the microgrid described by \eqref{eq:linear_microgrid} is small-signal stable. This is because the conditions \eqref{eq:Lyap_derivative_transformed} translate to \eqref{eq:Lyapunov_func_derivative} by defining $\tilde{\Psi} = T_1^{\top} \Psi T_1 \succ \pmb{0}$ (since $T_1$ is full rank) and $\tilde{\Pi} = T_2^{\top} \Pi T_2 \succeq \pmb{0}$. Before moving on to deriving the stability conditions, let us note that, if $A$ is a positive semidefinite matrix, then the following inequality holds for every pair of real-valued vectors $p$ and $q$:
\begin{equation}
  2 p^{\top} A q \geq -p^{\top} A p-q^{\top} A q.    
  \label{eq:quad_ineq}  
\end{equation}
%
\subsection{Conditions for Existence of $\Psi>0$ and $\Pi \geq 0$}
In a similar approach to that in \cite{vorobev2017framework}, let us choose the following transformation matrices.
\begin{equation}
    \begin{aligned}
    T_1 = & \begin{bmatrix} \pmb{I}_{3N} & \pmb{0} & \pmb{0} \\ \pmb{0} & \pmb{I}_{3N} & \pmb{0} \\  \pmb{I}_{3N} & \pmb{0} & 2 \tau \pmb{I}_{3N}  \end{bmatrix} \\
    \mbox{and}\; T_2 = & \begin{bmatrix} \pmb{I}_{3N} & \pmb{0} & \pmb{0} \\
    \Gamma(\Lambda_q + B + \tilde{B}) & -\Gamma G & \Gamma G^{\prime} \\
    \pmb{0} & \pmb{I}_{3N} & \pmb{0} \\
    \pmb{0} & \pmb{0} & \tau \pmb{I}_{3N}\end{bmatrix}
    \end{aligned}
    \label{eq:transformations}
\end{equation}
where $\Gamma = \left(B^{\prime}/\tau - \Lambda_q\right)^{-1}$, $\pmb{I}_{3N}$ is the $3N\times 3N$ identity matrix, and $\pmb{0}$ denotes matrices of zeros. Clearly, the matrix $T_1$ is full rank. Using \eqref{eq:Lyap_derivative_transformed} and \eqref{eq:transformations}, a Lyapunov function candidate can be obtained as $W=y^{\top}\Psi y$ where $y=T_1 x$ is the transformed state vector and $\Psi$ is a symmetric block-diagonal matrix defined as
\begin{equation}
    \begin{aligned}
        \Psi = & \begin{bmatrix} \Psi_{11} & \pmb{0} & \pmb{0} \\\pmb{0} & \Psi_{22} & \pmb{0} \\ \pmb{0} & \pmb{0} & \Psi_{33} \end{bmatrix} \\
        \mbox{where}\; \Psi_{11} = & \Lambda_p/2-B^{\prime} + 2 \tau B, \\
        \Psi_{22} = & 3 \tau \Lambda_q - B^{\prime} + 2 \tau (B+\tilde{B}), \\
        \Psi_{33} = & \Lambda_p/2.
    \end{aligned}
    \label{eq:psi_matrix}
\end{equation}
Note that the diagonal block $\Psi_{33}$ is positive definite since $\Lambda_p$ is a diagonal matrix with positive entries. Therefore, the condition that $\Psi \succ 0$ is equivalent to saying that both $\Psi_{11} \succ 0$ and $\Psi_{22} \succ 0$. As the next step, the time derivative ($\dot{W}$) of the Lyapunov function candidate along the trajectories of the system is computed. After some algebraic manipulation involving (1) multiplication of the first set of equations in \eqref{eq:linear_microgrid} separately by $\phi^{\top}$ and $2\tau \dot{\phi}^{\top}$; (2) multiplication of the second set of equations in \eqref{eq:linear_microgrid} separately by $\rho^{\top}$ and $2\tau \dot{\rho}^{\top}$ (3) addition of the resulting equations; and (4) rearrangement of some of the terms, the time derivative of the candidate Lyapunov function is given by 
\begin{equation}
    \begin{aligned}
        \dot{W} = & \frac{d}{dt} y^{\top} \Psi y = 2 y^{\top} \Psi\dot{y} = -2z^{\top} \hat{\Pi}z \\
        \mbox{where}\; \hat{\Pi} = &\begin{bmatrix}\hat{\Pi}_{11} & \hat{\Pi}_{12} & \hat{\Pi}_{13} & \pmb{0}\\ \hat{\Pi}_{12} & \hat{\Pi}_{22} & \pmb{0} & \pmb{0} \\ \hat{\Pi}_{13} & \pmb{0} & \hat{\Pi}_{33} & \hat{\Pi}_{34} \\ \pmb{0} & \pmb{0} & \hat{\Pi}_{34} & \hat{\Pi}_{44}\end{bmatrix} \\
        \hat{\Pi}_{11} = &B, \; \hat{\Pi}_{12} = -G-G^{\prime}/(2\tau), \; \hat{\Pi}_{13} = \tilde{G}/2 \\
        \hat{\Pi}_{22} = &2\Lambda_q-2B^{\prime}/\tau, \; \hat{\Pi}_{33} = \Lambda_q+B+\tilde{B} \\
        \hat{\Pi}_{34} = &G^{\prime}/(2\tau)+G+\tilde{G},\; \hat{\Pi}_{44} = (\Lambda_p-2B^{\prime})/\tau.
    \end{aligned}
    \label{eq:Lyap_derivative}
\end{equation}
Recalling \eqref{eq:matrices_definition}, $\hat{\Pi}_{13}\succeq \pmb{0}$. Moreover, it is assumed that $\hat{\Pi}_{34} \succeq \pmb{0}$; (recall that $G+\tilde{G} \succ \pmb{0}$ by definition, while $G^{\prime}$ is sign-indefinite). Applying the inequality \eqref{eq:quad_ineq} to $\hat{\Pi}$, we obtain
\begin{equation}
   \resizebox{0.09\hsize}{!}{%
   $\hat{\Pi} \succeq \Pi $
   }
    := 
    \resizebox{0.74\hsize}{!}{%
    $\begin{bmatrix} \begin{bmatrix} \hat{\Pi}_{11}-\hat{\Pi}_{13} & \hat{\Pi}_{12} \\ \hat{\Pi}_{12} & \hat{\Pi}_{22} \end{bmatrix} & \pmb{0} & \pmb{0} \\
    \pmb{0} & \hat{\Pi}_{33}-\hat{\Pi}_{13}-\hat{\Pi}_{34} & \pmb{0} \\ \pmb{0} & \pmb{0} & \hat{\Pi}_{44}-\hat{\Pi}_{34} \end{bmatrix}$
    }
    \label{eq:pi_ineq}
\end{equation}
where $\Pi$ has a block-diagonal structure. Therefore, $\Pi \succeq \pmb{0}$ if each of the diagonal blocks is also positive semidefinite. The conditions for existence of $\Psi \succ \pmb{0}$ and $\Pi \succeq \pmb{0}$ can therefore be summarized as
\begin{equation}
    \begin{aligned}
        & \Psi_{11} \succ \pmb{0}, \; \Psi_{22} \succ \pmb{0}, \; \begin{bmatrix} \hat{\Pi}_{11}-\hat{\Pi}_{13} & \hat{\Pi}_{12} \\ \hat{\Pi}_{12} & \hat{\Pi}_{22} \end{bmatrix} \succeq \pmb{0} \\
        & \hat{\Pi}_{33} - \hat{\Pi}_{13} - \hat{\Pi}_{34} \succeq \pmb{0} \;\; \mbox{and} \;\; \hat{\Pi}_{44} - \hat{\Pi}_{34} \succeq \pmb{0}.
    \end{aligned}
    \label{eq:psi_definition}
\end{equation}

However, algorithms to check these matrix conditions (such as Cholesky decomposition \cite{householder2013theory}) scale poorly with the size of the network. Moreover, having such network-wide conditions is not entirely desirable, since even local changes in the network (e.g., a new inverter, change in load, or a change in line) require the conditions to be reevaluated. It is, therefore, both efficient and useful to split the matrix inequalities into several sub-inequalities, each corresponding to a pair of inverter buses, similar to \cite{vorobev2017framework}. In particular, for each of the $N\times N$ matrices $\Lambda_p, \Lambda_q, B, G, \tilde{B}, \tilde{G}, B^{\prime}$ and $G^{\prime}$, and for every neighboring pair of inverter buses $(i,k)$, the following matrices are defined:
\begin{equation}
    \left[X\right]_{ik} := \begin{bmatrix} \frac{X_{ii}+\sum_{j=1}^{n_i} X_{ij}}{n_i} - X_{ik} & X_{ik} \\ X_{ki} & \frac{X_{kk}+\sum_{j=1}^{n_k} X_{kj}}{n_k} - X_{ki} \end{bmatrix}
    \label{eq:matrix_partition}
\end{equation}
where X is used to represent any of $\Lambda_p, \Lambda_q, B, G, \tilde{B}, \tilde{G}, B^{\prime}$ and $G^{\prime}$; $X_{ik}$ (for every $i$ and $k$) denotes the $i$-th row and $k$-th column entry in $X$; $n_i$ and $n_k$ denote the number of neighbors of bus $i$ and bus $k$, respectively. Finally, the conditions of small-signal stability of the microgrid are summarized as follows:
\begin{theorem}
\label{thm:main}
The microgrid network described by \eqref{eq:linear_microgrid} is small-signal stable if the following inequalities hold for every pair of neighboring buses $(i,k)$:
\begin{equation}
    \begin{aligned}
        & [G]_{ik} + [\tilde{G}]_{ik} + [G^{\prime}]_{ik}/2\tau \succeq \pmb{0} \\
        & [\Lambda_p]_{ik}/2 - [B^{\prime}]_{ik} + 2\tau[B]_{ik} \succ \pmb{0} \\
        & [\Lambda_p]_{ik} - 2 [B^{\prime}]_{ik} - [G^{\prime}]_{ik}/2 - \tau [G]_{ik} - \tau [\tilde{G}]_{ik} \succeq \pmb{0} \\
        & [\Lambda_q]_{ik} + [B]_{ik} + [\tilde{B}]_{ik} \succeq \frac{3}{2} [\tilde{G}]_{ik} + \frac{1}{2\tau} [G^{\prime}]_{ik} + [G]_{ik} \\
        & [\Lambda_q]_{ik} - [B]_{ik}/\tau \succ \pmb{0} \; \mbox{and} \\
        & 2 [B]_{ik} - [\tilde{G}]_{ik} \succeq \\
        &  
           \resizebox{0.85\hsize}{!}{%
    $\left([G]_{ik} + \frac{[G^{\prime}]_{ik}}{2\tau}  \right) \left([\Lambda_q]_{ik}  - \frac{[B^{\prime}]_{ik} }{\tau} \right)^{-1} \left([G]_{ik} + \frac{[G^{\prime}]_{ik}}{2\tau} \right)$
    }
    \end{aligned}
    \label{eq:stability_conditions}
\end{equation}
\end{theorem}
\begin{proof}
Recall that the conditions in \eqref{eq:psi_definition}, together with the condition that $\hat{\Pi}_{34} = G^{\prime}/(2\tau)+G+\tilde{G}\succeq \pmb{0}$, are sufficient for $\Psi \succ \pmb{0}$ and $\Pi \succeq \pmb{0}$. Also recall that $\hat{\Pi}_{22} = 2\Lambda_q-2B^{\prime}/\tau$ needs to be nonsingular for the existence of the transformation matrix $T_2$. Therefore, taking Schur's complement, the condition 
\[\begin{bmatrix} \hat{\Pi}_{11}-\hat{\Pi}_{13} & \hat{\Pi}_{12} \\ \hat{\Pi}_{12} & \hat{\Pi}_{22} \end{bmatrix} \succeq \pmb{0} \] 
is expressed as two matrix inequalities: $\hat{\Pi} \succ \pmb{0}$ and $\hat{\Pi}_{11} - \hat{\Pi}_{13} - \hat{\Pi}_{12} \hat{\Pi}_{22}^{-1} \hat{\Pi}_{12} \succeq \pmb{0}$. Moreover, with a rearrangement of terms, it is easy to see that $\hat{\Pi}_{22} \succ \pmb{0}$ also implies $\Psi_{22} \succ \pmb{0}$. Finally, note that the matrices $\Lambda_p, \Lambda_q, B, G, \tilde{B}, \tilde{G}, B^{\prime}$ and $G^{\prime}$ are either diagonal or have the same sparsity structure as the physical network. Therefore, each of those can be expressed as a sum of $M$ ($M$ is the number of lines) sparse matrices, each corresponding to a neighboring pair of inverter buses, where the only nonzero entries correspond to the rows and columns of the corresponding pair of buses. Enforcing the positive (semi)definiteness on each of the constituent submatrices corresponding to every pair of neighboring buses $(i,k)$, the stricter stability conditions in \eqref{eq:stability_conditions} are obtained, having a computational complexity that scales linearly with the number of neighboring pairs of inverter buses.						
\end{proof}

Note that the stability conditions for any pair of non-inverter buses in Eq. \eqref{eq:stability_conditions} are of size $6\times 6$ and for any pair of inverter buses, they are of size $2\times 2$. In particular, a $2\times 2$ matrix is positive definite (or, semidefinite) if its determinant and either of the diagonal entries are positive (or, nonnegative). Thus, the matrix conditions in \eqref{eq:stability_conditions} for a pair of inverter buses yield  (efficiently verifiable) algebraic conditions on the network and inverter droop parameters, as demonstrated via numerical examples. 

\subsection{Sensitivity}
\label{sec:sensitivity}
The inequality conditions in \eqref{eq:stability_conditions} could be analyzed to derive sensitivity relations between the stability boundary (in droop parameters) and network parameters (e.g. line impedance values). In order to see some of those relations analytically, consider the specific scenario when the line resistance and reactance values are varied proportionally, by some positive scaling factor $\alpha$. In such a case, the matrices $B, G, \tilde{B}, \tilde{G}, B^{\prime}$ and $G^{\prime}$ are scaled by the factor $1/\alpha$. Therefore, all of the stability conditions in \eqref{eq:stability_conditions} would be satisfied if the matrices $\Lambda_p$ and $\Lambda_q$ are also scaled by $1/\alpha$, which implies scaling the droop coefficients by the factor $\alpha$. Thus, it can be argued that an increase in the line impedance values results in larger estimates of the stability region in droop coefficients, and \textit{vice versa}. This analytical observation is also validated via numerical simulations as presented in the next section.

\section{Numerical Results}
\label{sec:results}
In this section, the proposed approach is illustrated on two different microgrid systems: a two-bus (balanced) system, and the (modified) IEEE 34-bus system. The two-bus system is used to demonstrate how the stability region changes in presence of loads; while the IEEE 34-bus system is used to validate the analytical results on realistic network models and to investigate the effects of various design choices (e.g. conductor types, inverter sizes) on the stability regions. 
\begin{figure}[h!]
    \centering
    \includegraphics[scale = 0.51]{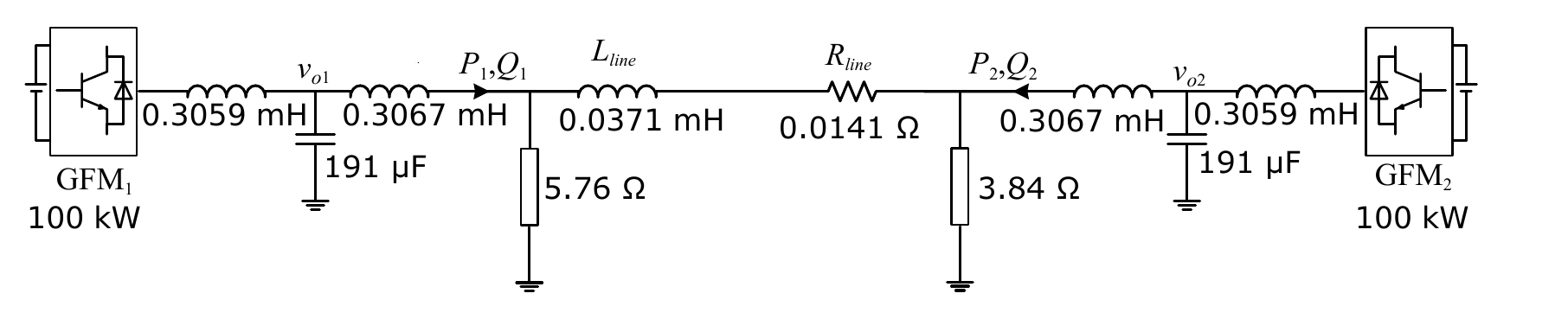}
    \caption{Two bus system with GFMs and loads.}
    \label{fig:two_bus_sys}
\end{figure}

The two-bus system consists of two grid-forming inverters (denoted by `GFMs') as shown in Fig.  \ref{fig:two_bus_sys}. The stability regions with respect to the inverters' \textit{P-f} and \textit{Q-V} droop parameters are generated for two cases: with and without loads, using the analytical stability conditions from Theorem \ref{thm:main} and \cite{vorobev2017framework}, respectively.

In Fig. \ref{fig:two_bus_simulations}, the plots in dotted lines are obtained using analytical stability conditions and the solid lines are obtained from eigenvalue analysis. Fig. \ref{fig:two_bus_simulations} shows that the stability region with loads will be smaller than the stability region without loads. It is clear that the estimated stability boundary without load is not completely contained in the eigenvalue-based stability boundary with load. This leads to an overestimate of the stability regions and might result in incorrect conclusions.
\begin{figure}[h!]
    \centering
    \includegraphics[width = 0.9 \columnwidth]{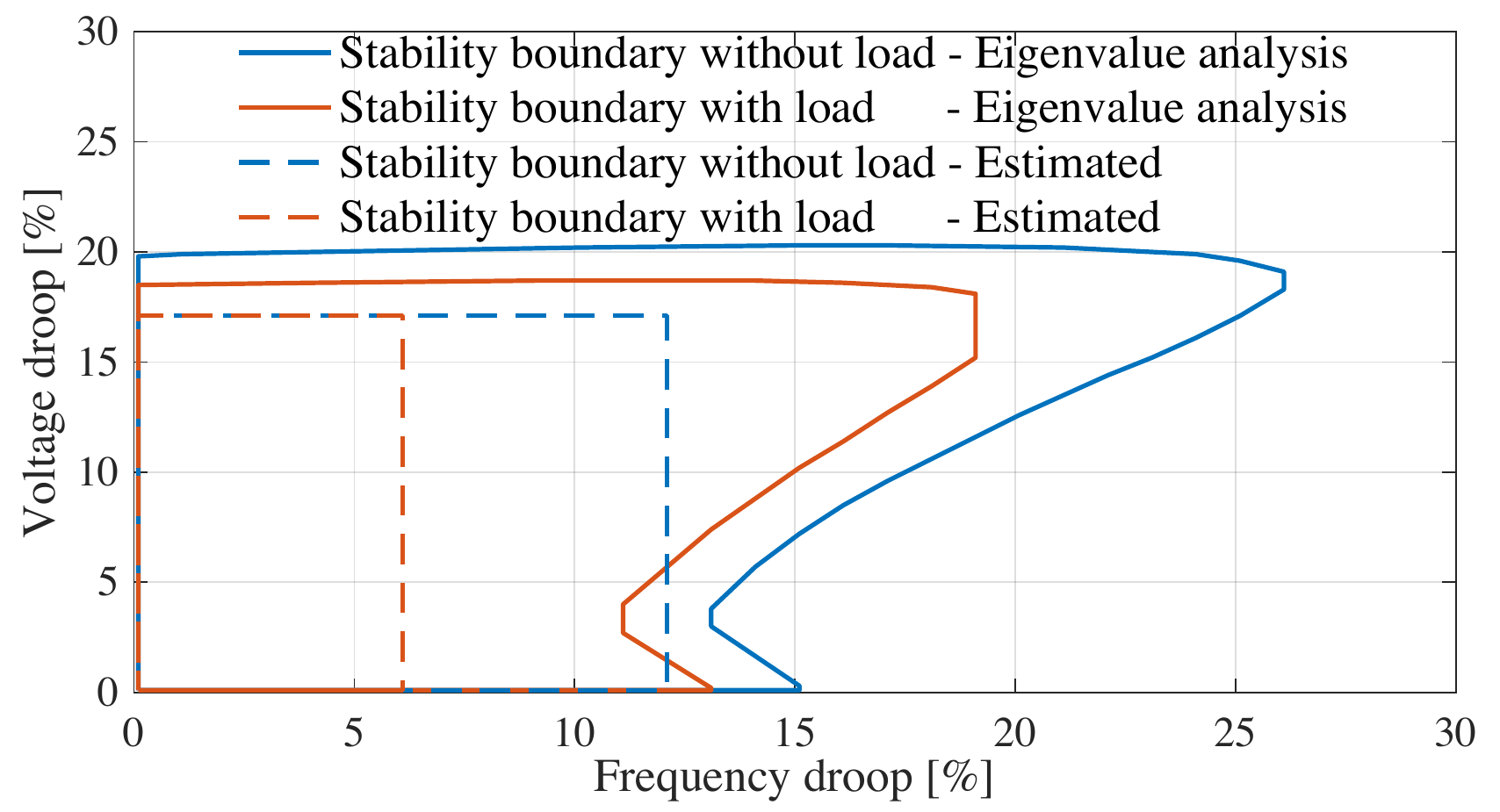}
    \caption{Eigenvalue-based and estimated stability boundaries for a two-bus system considered with and without loads.}
    \label{fig:two_bus_simulations}
\end{figure}

The proposed methodology to study small-signal stability is next illustrated on the IEEE 34-bus system, modified to incorporate inverters, creating a microgrid. In particular, the three-phase unbalanced 34-bus system is considered with 10 GFMs. The one-line diagram for the 34-bus system with GFMs is shown in Fig. \ref{fig:34bus_system}. The line between Buses 800 and 802 is not considered here, as it is assumed that there are inverters installed to meet the load demand of the 34-bus system (i.e., it is an islanded system and not connected to a bulk power system). The network and load data for the 34-bus system are obtained from \cite{ieeetestfeeder}. The locations and sizes of the inverters are 
as follows - Bus 806: 300 kW, Bus 816: 300 kW, Bus 828: 200 kW, Bus 830: 150 kW, Bus 836: 60 kW, Bus 840: 300 kW, Bus 842: 50 kW, Bus 844: 100 kW, Bus 846: 120 kW, and Bus 860: 150 kW. 
\begin{figure}[h!]
    \centering
    \includegraphics[width = 1 \columnwidth]{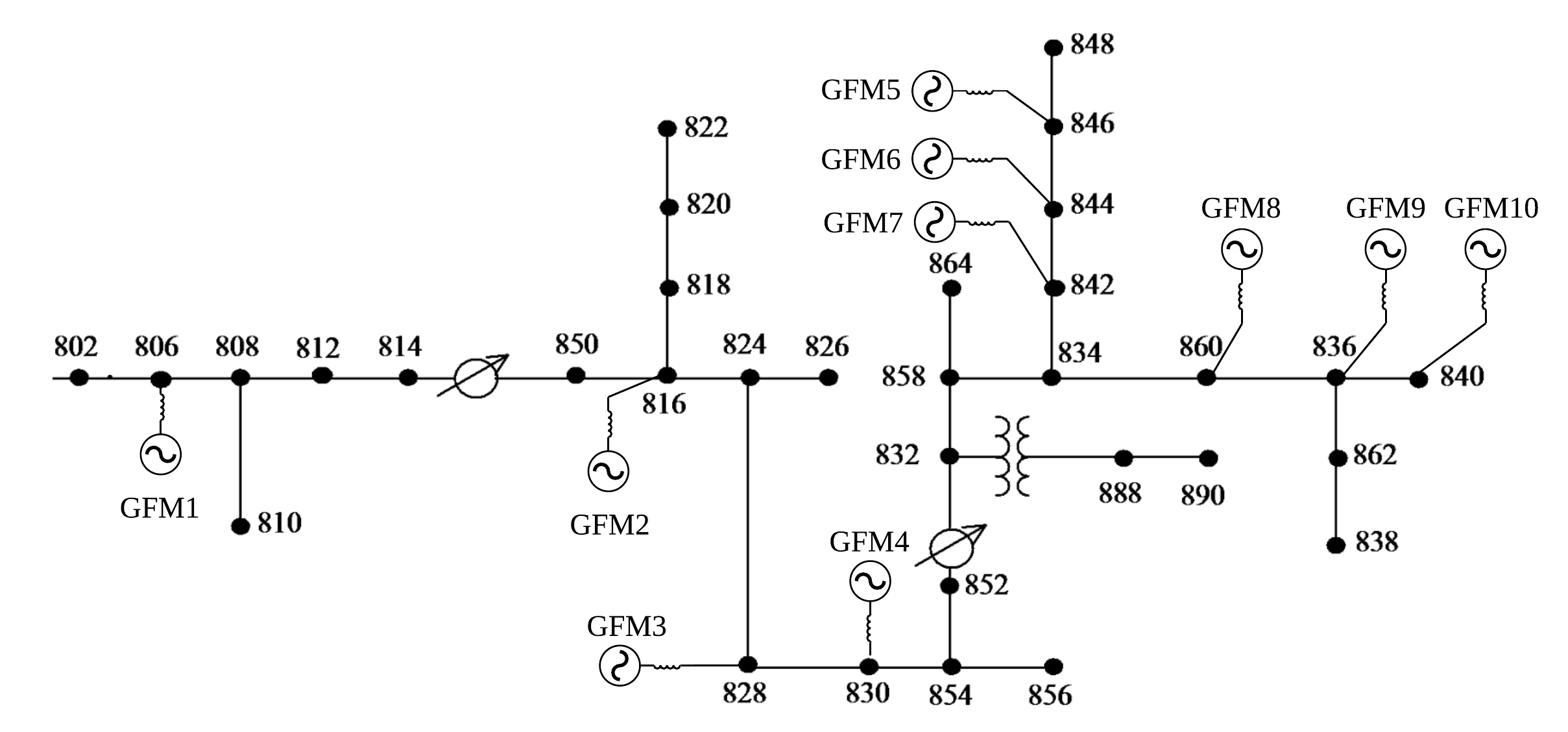}
    \caption{IEEE 34-bus network, a three-phase unbalanced system.}
    \label{fig:34bus_system}
    \vspace{-0.5 cm}
\end{figure}
\subsection{Effect of Conductor Type}
Various conductor types that have different conductor strands and current ratings are considered here for the analysis. The original conductor type mentioned in the IEEE 34-bus system data file is \textbf{C2}: \#2 6/1 ACSR 180(A). Four different types of conductors with different sizes, types of strands, and capacities for the same conductor material are considered, namely, \textbf{C1}: \#4 7/1 ACSR 140(A), \textbf{C3}: \#1/0 ACSR 230(A), \textbf{C4}: \#4/0 6/1 ACSR 340(A), and \textbf{C5}: \#336,400 26/7 ACSR 530(A) \cite{kersting2006distribution}. Changing the conductor type changes the effective line impedance of the lines connecting the buses. Conductor types are considered in descending order of both resistance (R) and reactance (X) values and in ascending order with respect to current capacity ratings. 

The stability boundaries between Buses 828 and 830 with respect to different conductor types are computed by applying the \textit{P-f} stability condition given in Theorem 1 and validated against eigenvalue analysis of \eqref{eq:linear_microgrid} as well as PSCAD \cite{pscad} simulations based on the full electromagnetic model. From Fig. \ref{fig:Pf_vs_CT_SB}, it can be seen that the eigenvalue-based stability boundary is a good approximation of the PSCAD-based stability boundary. Hence, it should be adequate to compare the analytically based stability regions against eigenvalue analysis results. Further, the stability boundary obtained by applying the sufficient conditions is an inner approximate, that is, strictly below the stability boundaries computed from eigenvalue analysis and PSCAD simulations. While the stability boundaries computed from PSCAD simulations and eigenvalue analysis both remain relatively unchanged, the analytically estimated stability region shrinks as the conductor type is changed from \textbf{C1} to \textbf{C5}, in a descending order of the effective line impedance values. Moreover, applying the discussion in Section \ref{sec:sensitivity}, it can be clearly seen how the reduction in effective line impedance affects the stability region estimates. 
\begin{figure}[h!]
    \centering
    \includegraphics[width = 0.9 \columnwidth]{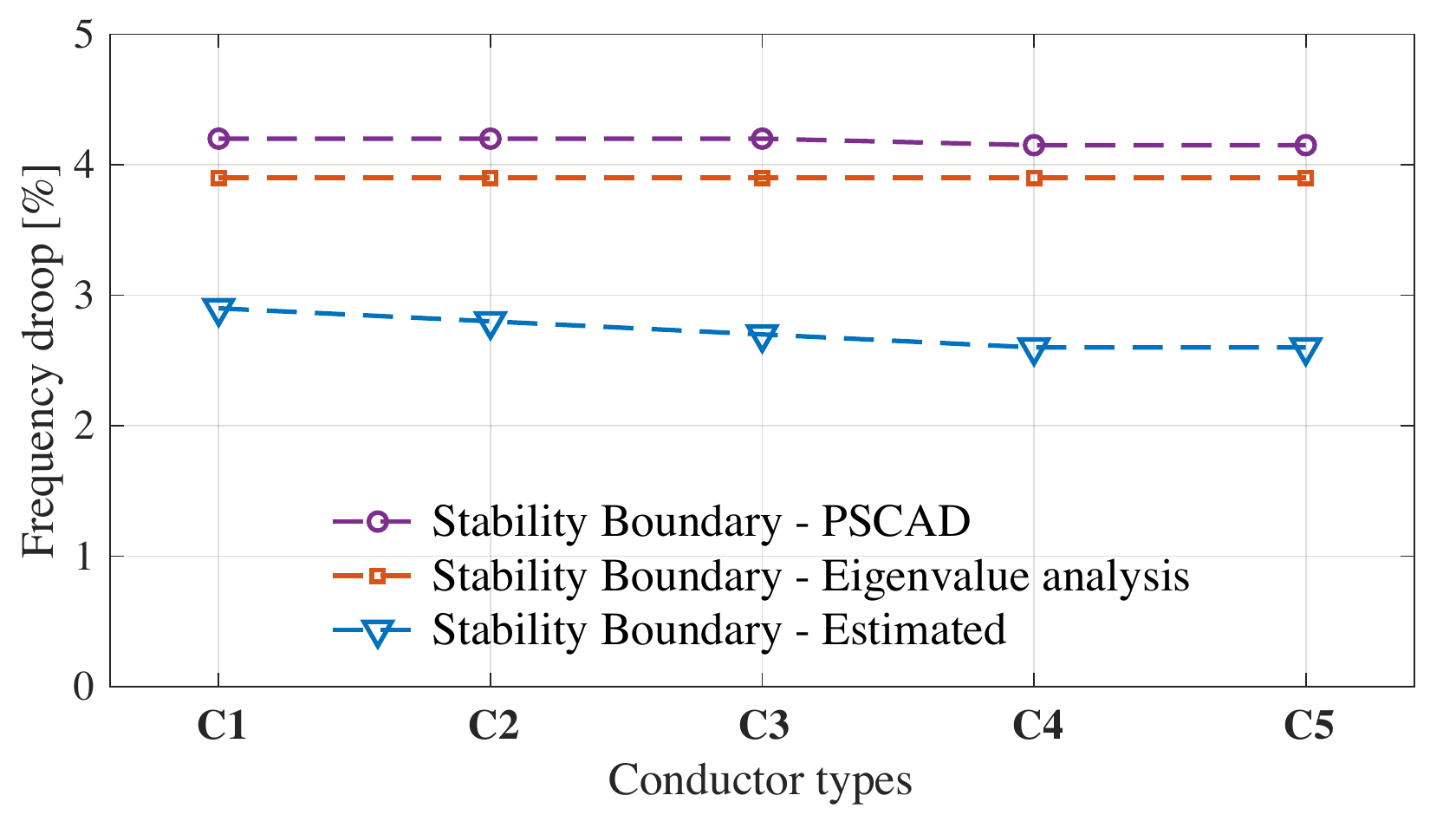}
    \caption{Stability boundary with respect to P-f droop between inverter buses 828-830 with different choice of conductors. Nominal line length is 20440 ft.}
    \label{fig:Pf_vs_CT_SB}
\end{figure}

\begin{figure}[h!]
    \centering
    \includegraphics[scale = 0.45]{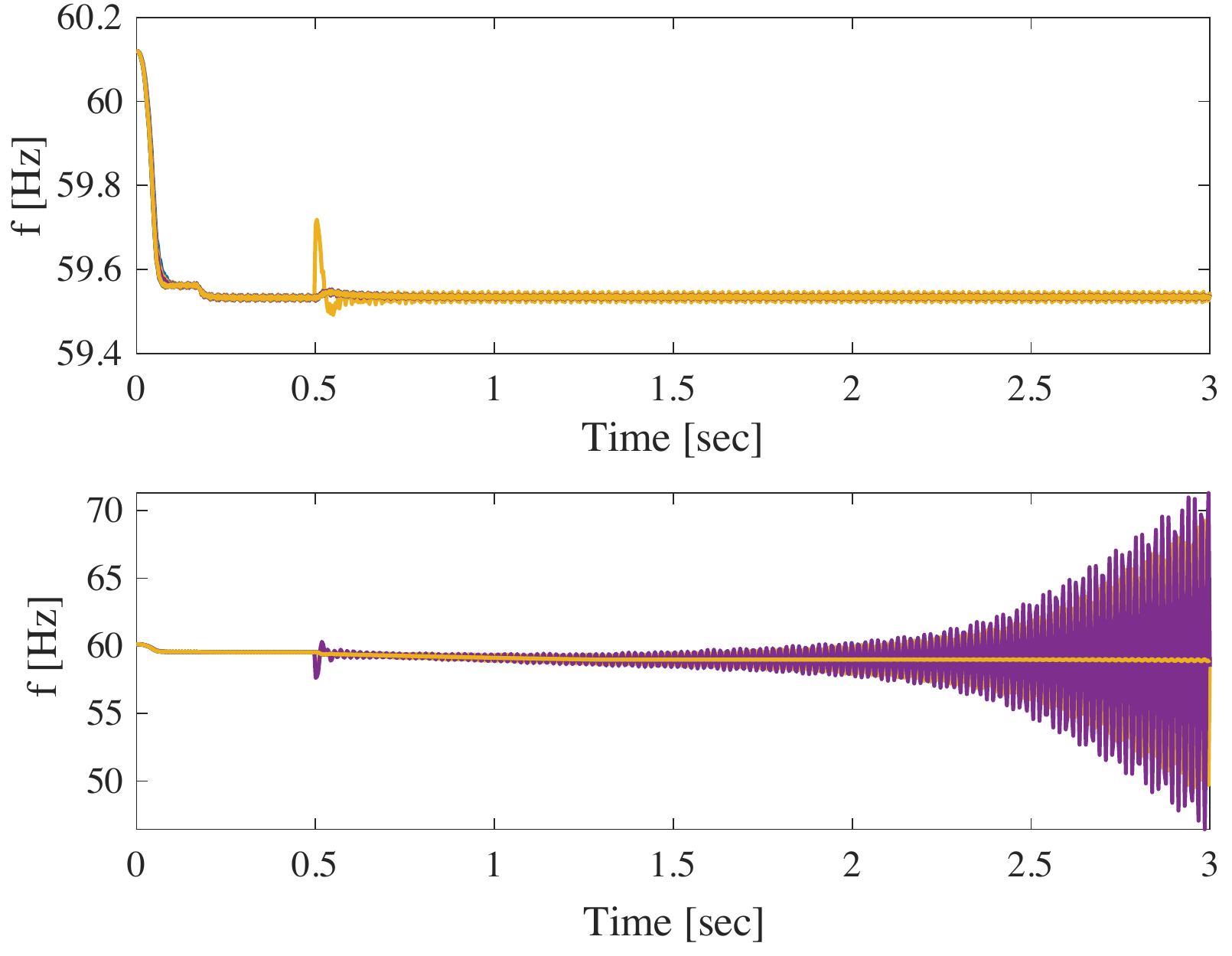}
    \caption{PSCAD simulations (corresponding to conductor type \textbf{C2}) showing stable (top) \& unstable (bottom) scenarios. In both cases, generation at Bus 828 increases by 0.3 p.u. at 0.5s. Corresponding \textit{P-f} droop gains for stable and unstable scenarios are 1\% and 5\%.}
    \label{fig:pscad_simulations}
\end{figure}

In PSCAD simulations, the stability is determined by monitoring the frequencies, and the active and reactive power generation in the presence of a disturbance. For illustration purposes, generation at one inverter (Bus 828) is increased by 0.3 p.u. at 0.5 seconds and the \textit{P-f} droop coefficient is gradually increased to identify the stability boundary. Fig.  \ref{fig:pscad_simulations} shows the frequencies at all inverter buses for both stable and unstable scenarios. When the \textit{P-f} droop gains are increased beyond 4.2\%, the system exhibits instability; the corresponding plots are shown in Fig. \ref{fig:pscad_simulations}. 

Note that, while verifying the small-signal stability using PSCAD simulations and eigenvalue analysis, for any changes in any part of the network, the PSCAD simulations have to be redone for the entire system. Furthermore, the system matrices for the entire system has to be recomputed for eigenvalue analysis. However, using the proposed distributed stability conditions, it is enough to just verify the small-signal stability corresponding to the location where changes have occurred. This shows the advantage of proposed distributed small-signal stability conditions over eigenvalue analysis and PSCAD simulations.

Let us denote the maximum \textit{P-f} droop values obtained via analytical estimates and the eigenvalue analysis as $\lambda_{p,est}^*$ and $\lambda_{p,eigen}^*$, respectively. Clearly, it is observed that $\lambda_{p,est}^*<\lambda_{p,eigen}^*$, i.e. the analytical stability region is a conservative estimate of the eigenvalue analysis-based stability region. The degree of conservativeness can be quantitatively measured by
\begin{equation}
    \frac{\lambda_{p,sim}^* - \lambda_{p,est}^*}{\lambda_{p,sim}^*} \times 100 \%.
    \label{eq:conservativeness}
\end{equation}
A heat-map is generated in Fig. \ref{fig:heat_map_CT} to illustrate the degree of conservativeness of the analytical estimate of the stability boundary, in comparison to the eigenvalue analysis-based estimate of the stability boundary. Conservativeness of the estimated stability boundaries vary with respect to the conductor types and bus locations. In particular, for any given pair of buses, the analytical estimate is observed to be increasingly more conservative as the conductor type is changed from \textbf{C1} to \textbf{C5}, in a descending order of the effective line impedance values. Note that the results of the row corresponding to the pair of buses 828 and 830 agree with the values in Fig. \ref{fig:heat_map_CT}.
\begin{figure}[h!]
    \centering
    \includegraphics[width = 0.9 \columnwidth]{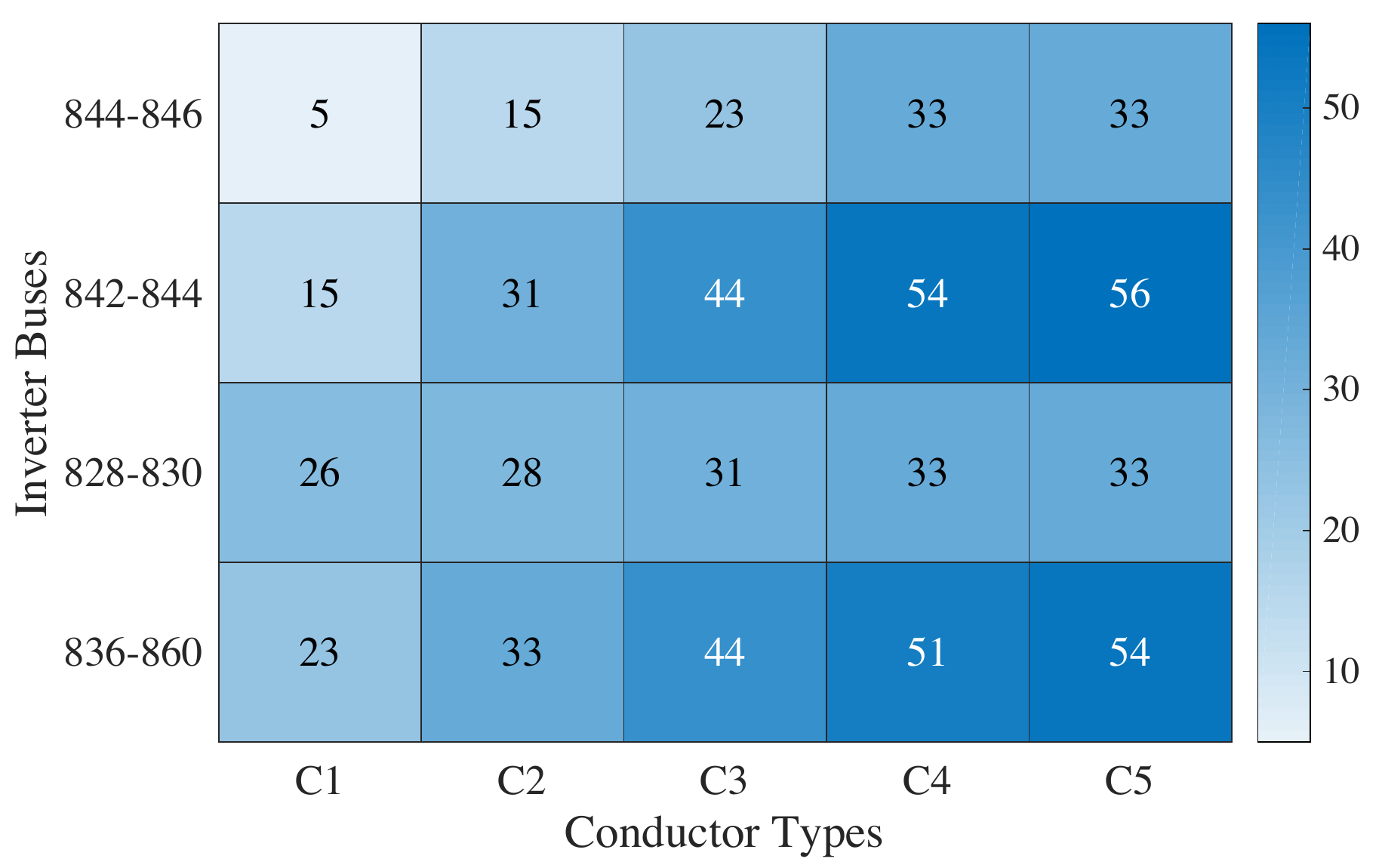}
    \caption{Percentage of conservativeness between the eigenvalue-based stability boundary and the estimated stability boundary applying analytical (sufficient) stability conditions by choosing different conductor types at various locations.}
    \label{fig:heat_map_CT}
\end{figure}

\subsection{Effect of Inverter Size}
In this subsection, applying the proposed analytical results, the effect of inverter size and location of inverter buses on the stability margins are studied. Let us consider the pair of Buses 828 and 830 with inverters. The inverter size at Bus 828 is kept constant at 200 kW, while the inverter size at Bus 830 is varied over from 120 kW to 180 kW. Fig.  \ref{fig:IS} shows the stability region between Buses 828 and 830 with respect to frequency droop and inverter size. The stability region obtained from eigenvalue analysis agrees closely with the PSCAD-based stability region. Although the stability boundaries computed from PSCAD simulations and eigenvalue analysis remain relatively unchanged, the analytically estimated stability region shrinks as the inverter size is increased. Moreover, the inverter size and location also play a significant role in determining the conservativeness of analytically based stability regions, as illustrated by the heat-map in Fig. \ref{fig:heat_map_IS}. In particular, the conservativeness of the analytical estimate increases with the increase in the inverter size. However, the degree of conservativeness is also sensitive to the location of the inverter. For instance, the sensitivity of the degree of conservativeness to the inverter size is relatively high on buses 830 and 846, while that on the buses 828 and 844 is relatively low. 

Fig. \ref{fig:heat_map_CT} and Fig. \ref{fig:heat_map_IS} best exemplifies the applicability of the proposed distributed stability conditions. It can be seen in these figures, we can estimate the stability region with respect to any pair of buses.
\begin{figure}[h!]
    \centering
    \includegraphics[width = 0.9 \columnwidth]{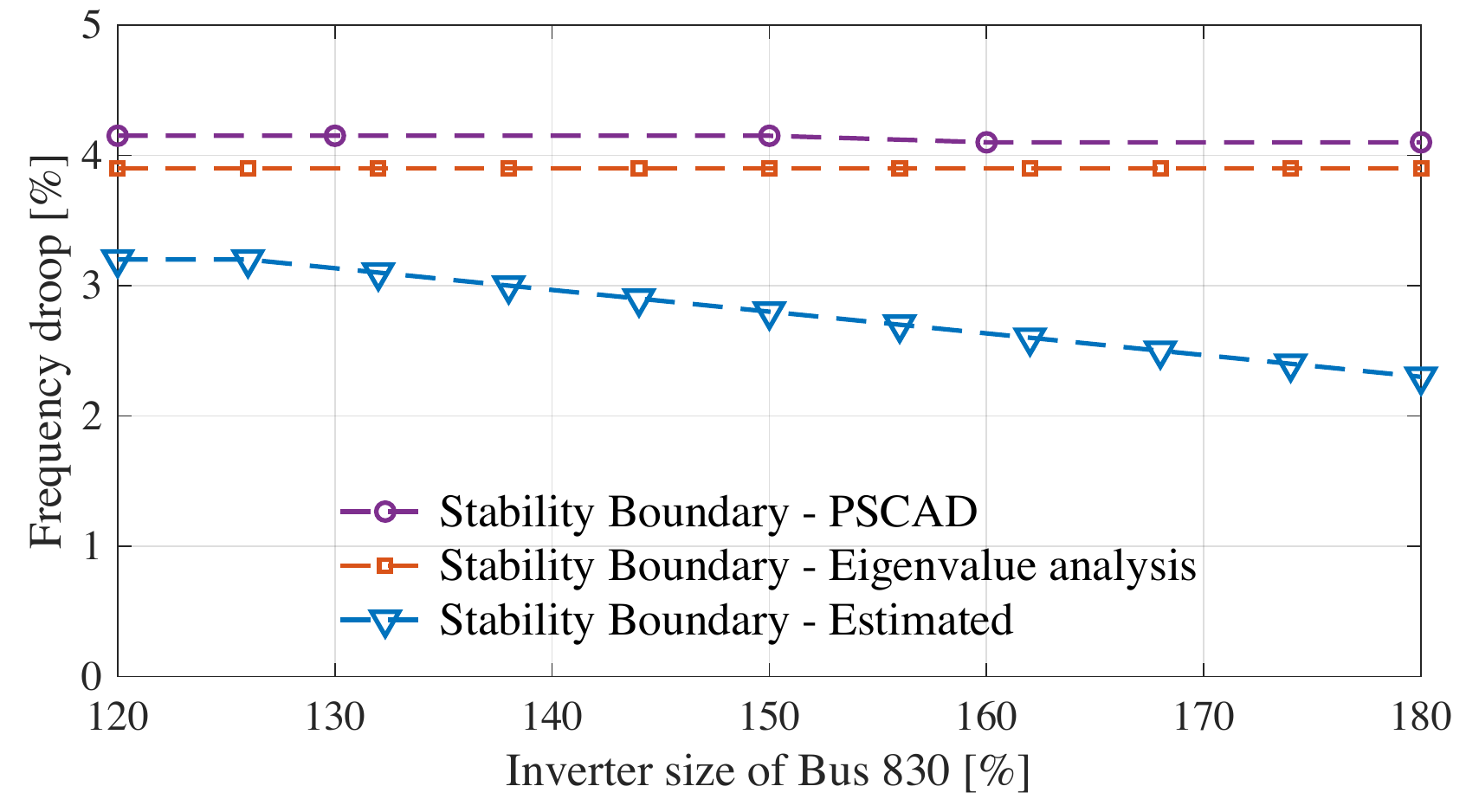}
    \caption{\textit{P-f} based stability boundary between Buses 828 and 830 when the inverter size at Bus 830 is varied. Nominal line length is 20440 ft.}
    \label{fig:IS}
\end{figure}

\begin{figure}[h!]
    \centering
    \includegraphics[width = 0.88 \columnwidth]{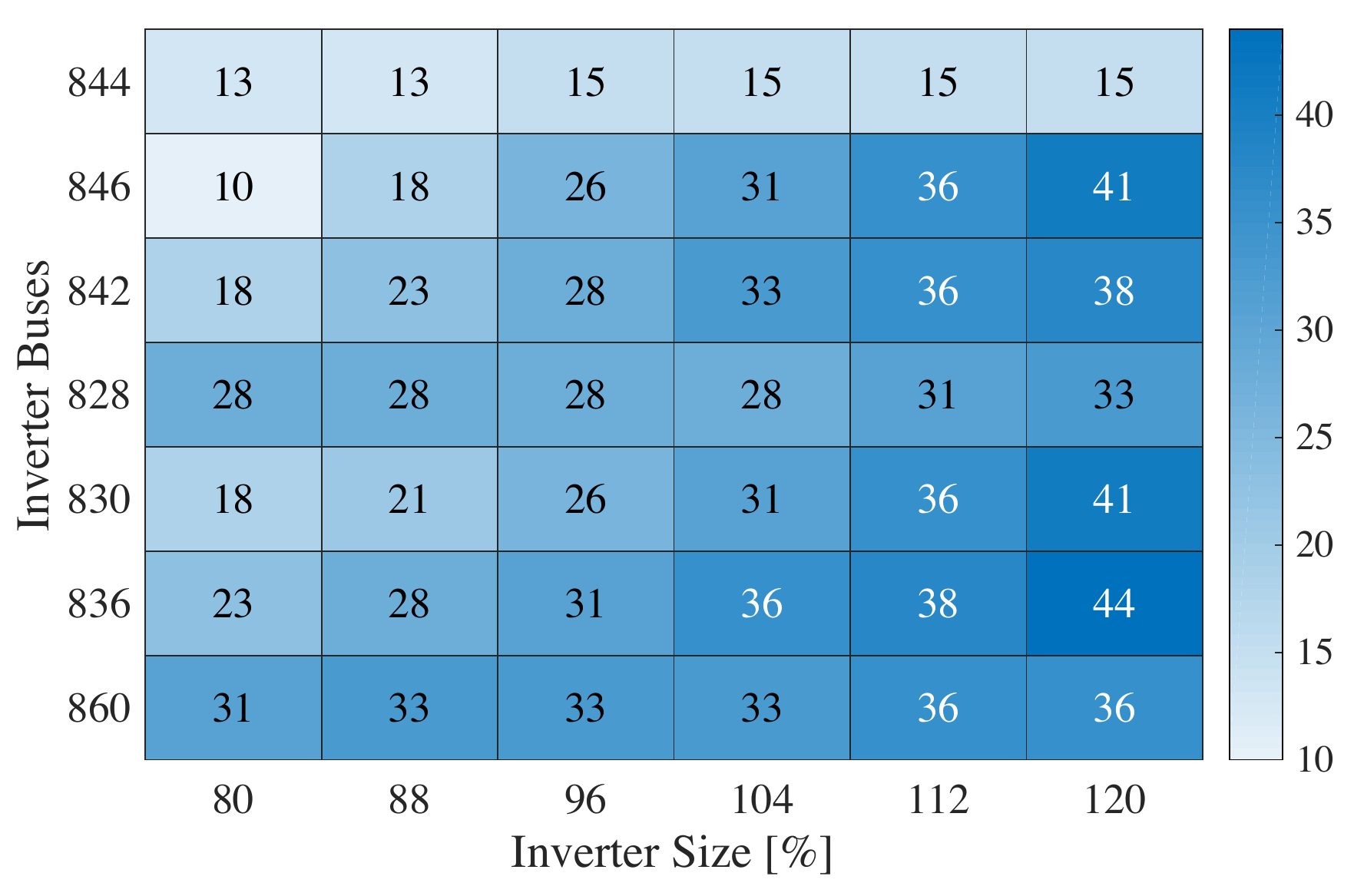}
    \caption{Percentage of conservativeness between the eigenvalue-based stability boundary and the estimated stability boundary applying analytical (sufficient) stability conditions for a range of inverter sizes ($\pm 20$\% from nominal) at various locations.}
    \label{fig:heat_map_IS}
\end{figure}

\section{Conclusion}
\label{sec:conclusion}

The distributed analytical stability conditions derived in this work scale linearly with the number of lines in the network and can be verified in a computationally efficient manner. Closed-form distributed sufficient conditions for small-signal stability are obtained using realistic system models, that are validated using PSCAD-based electromagnetic simulations. The proposed method allows the consideration of modeling complexities such as the electromagnetic transient effects of inverter-based systems and the three-phase, unbalanced, lossy microgrid networks.

PSCAD-based detailed electromagnetic simulations are used to validate the stability conditions on realistic test systems (e.g. IEEE 34-bus network). Moreover, studies on the sensitivities and conservativeness of the proposed analytical conditions with respect to various conductor types, as well as inverter sizes and locations are provided. In addition, introducing a metric for the degree of conservatives, heat-maps are generated to illustrate how conservative of the analytical results are under various scenarios: varying conductor types, inverter sizes, and inverter locations. Moreover, the proposed distributed small-signal stability conditions gives guarantees on stability regions although they provide conservative estimates. Future work will focus on analytical methods to further investigate the sensitivity of the stability regions to various design parameters, as well as on extending the current results to incorporate uncertainty in the network parameters.

\section*{Acknowledgment}
The authors offer their sincere thanks to Dr. Petr Vorobev (while he was at the Massachusetts Institute of Technology) and Dr. Long Vu, Dr. Kathleen Nowak, Dr. David Engel (from PNNL) for many helpful discussions, and Maura Zimmerman for proofreading the manuscript. We also thank the anonymous reviewers for their constructive comments on the manuscript.


%

\ifCLASSOPTIONcaptionsoff
  \newpage
\fi

\ifCLASSOPTIONcaptionsoff
  \newpage
\fi


\begin{thebibliography}{10}
\providecommand{\url}[1]{#1}
\csname url@rmstyle\endcsname
\providecommand{\newblock}{\relax}
\providecommand{\bibinfo}[2]{#2}
\providecommand\BIBentrySTDinterwordspacing{\spaceskip=0pt\relax}
\providecommand\BIBentryALTinterwordstretchfactor{4}
\providecommand\BIBentryALTinterwordspacing{\spaceskip=\fontdimen2\font plus
\BIBentryALTinterwordstretchfactor\fontdimen3\font minus
  \fontdimen4\font\relax}
\providecommand\BIBforeignlanguage[2]{{%
\expandafter\ifx\csname l@#1\endcsname\relax
\typeout{** WARNING: IEEEtran.bst: No hyphenation pattern has been}%
\typeout{** loaded for the language `#1'. Using the pattern for}%
\typeout{** the default language instead.}%
\else
\language=\csname l@#1\endcsname
\fi
#2}}

\bibitem{schneider2016evaluating}
K.~P. Schneider, F.~K. Tuffner, M.~A. Elizondo, C.-C. Liu, Y.~Xu, and D.~Ton,
  ``Evaluating the feasibility to use microgrids as a resiliency resource,''
  \emph{IEEE Transactions on Smart Grid}, vol.~8, no.~2, pp. 687--696, 2016.

\bibitem{schneider2019distributed}
K.~P. Schneider, S.~Laval, J.~Hansen, R.~B. Melton, L.~Ponder, L.~Fox, J.~Hart,
  J.~Hambrick, M.~Buckner, M.~Baggu, \emph{et~al.}, ``A distributed power
  system control architecture for improved distribution system resiliency,''
  \emph{IEEE Access}, vol.~7, pp. 9957--9970, 2019.

\bibitem{lasseter2011smart}
R.~H. Lasseter, ``Smart distribution: Coupled microgrids,'' \emph{Proceedings
  of the IEEE}, vol.~99, no.~6, pp. 1074--1082, 2011.

\bibitem{schneider2018enabling}
K.~P. Schneider, F.~K. Tuffner, M.~A. Elizondo, C.-C. Liu, Y.~Xu, S.~Backhaus,
  and D.~Ton, ``Enabling resiliency operations across multiple microgrids with
  grid friendly appliance controllers,'' \emph{IEEE Transactions on Smart
  Grid}, vol.~9, no.~5, pp. 4755--4764, 2018.

\bibitem{pogaku2007modeling}
N.~Pogaku, M.~Prodanovic, and T.~C. Green, ``Modeling, analysis and testing of
  autonomous operation of an inverter-based microgrid,'' \emph{IEEE
  Transactions on Power Electronics}, vol.~22, no.~2, pp. 613--625, 2007.

\bibitem{coelho2002small}
E.~A.~A. Coelho, P.~C. Cortizo, and P.~F.~D. Garcia, ``Small-signal stability
  for parallel-connected inverters in stand-alone ac supply systems,''
  \emph{IEEE Transactions on Industry Applications}, vol.~38, no.~2, pp.
  533--542, 2002.

\bibitem{du2014voltage}
W.~Du, Q.~Jiang, M.~J. Erickson, and R.~H. Lasseter, ``{Voltage-source control
  of PV inverter in a CERTS microgrid},'' \emph{IEEE Transactions on Power
  Delivery}, vol.~29, no.~4, pp. 1726--1734, 2014.

\bibitem{pushpak2015power}
S.~Pushpak, H.~Pota, and U.~Vaidya, ``Power sharing in microgrids with minimum
  communication control,'' in \emph{2015 IEEE Power \& Energy Society General
  Meeting}.\hskip 1em plus 0.5em minus 0.4em\relax IEEE, 2015, pp. 1--5.

\bibitem{simpson2013synchronization}
J.~W. Simpson-Porco, F.~D{\"o}rfler, and F.~Bullo, ``Synchronization and power
  sharing for droop-controlled inverters in islanded microgrids,''
  \emph{Automatica}, vol.~49, no.~9, pp. 2603--2611, 2013.

\bibitem{schiffer2014conditions}
J.~Schiffer, R.~Ortega, A.~Astolfi, J.~Raisch, and T.~Sezi, ``Conditions for
  stability of droop-controlled inverter-based microgrids,'' \emph{Automatica},
  vol.~50, no.~10, pp. 2457--2469, 2014.

\bibitem{vorobev2017framework}
P.~Vorobev, P.-H. Huang, M.~Al~Hosani, J.~L. Kirtley, and K.~Turitsyn, ``A
  framework for development of universal rules for microgrids stability and
  control,'' in \emph{2017 IEEE 56th Annual Conference on Decision and Control
  (CDC)}.\hskip 1em plus 0.5em minus 0.4em\relax IEEE, 2017, pp. 5125--5130.

\bibitem{kundu2019identifying}
S.~Kundu, W.~Du, S.~P. Nandanoori, F.~Tuffner, and K.~Schneider, ``Identifying
  parameter space for robust stability in nonlinear networks: A microgrid
  application,'' in \emph{2019 American Control Conference (ACC)}.\hskip 1em
  plus 0.5em minus 0.4em\relax IEEE, 2019.

\bibitem{kundu2019distributed}
S.~Kundu, S.~P. Nandanoori, K.~Kalsi, S.~Geng, and I.~A. Hiskens, ``Distributed
  barrier certificates for safe operation of inverter-based microgrids,'' in
  \emph{2019 American Control Conference (ACC)}.\hskip 1em plus 0.5em minus
  0.4em\relax IEEE, 2019, pp. 1042--1047.

\bibitem{Vorobev2019decentralized}
P.~Vorobev, S.~Chevalier, and K.~Turitsyn, ``Decentralized stability rules for
  microgrids,'' in \emph{2019 American Control Conference (ACC)}.\hskip 1em
  plus 0.5em minus 0.4em\relax IEEE, 2019.

\bibitem{kroposki2017achieving}
B.~Kroposki, B.~Johnson, Y.~Zhang, V.~Gevorgian, P.~Denholm, B.-M. Hodge, and
  B.~Hannegan, ``Achieving a 100\% renewable grid: Operating electric power
  systems with extremely high levels of variable renewable energy,'' \emph{IEEE
  Power and Energy Magazine}, vol.~15, no.~2, pp. 61--73, 2017.

\bibitem{chandorkar1993control}
M.~C. Chandorkar, D.~M. Divan, and R.~Adapa, ``Control of parallel connected
  inverters in standalone ac supply systems,'' \emph{IEEE Transactions on
  Industry Applications}, vol.~29, no.~1, pp. 136--143, 1993.

\bibitem{denis2018migrate}
G.~Denis, T.~Prevost, M.-S. Debry, F.~Xavier, X.~Guillaud, and A.~Menze, ``The
  migrate project: the challenges of operating a transmission grid with only
  inverter-based generation. a grid-forming control improvement with transient
  current-limiting control,'' \emph{IET Renewable Power Generation}, vol.~12,
  no.~5, pp. 523--529, 2018.

\bibitem{taylor2016power}
J.~A. Taylor, S.~V. Dhople, and D.~S. Callaway, ``Power systems without fuel,''
  \emph{Renewable and Sustainable Energy Reviews}, vol.~57, pp. 1322--1336,
  2016.

\bibitem{farrokhabadi2019microgrid}
M.~Farrokhabadi, C.~A. Ca{\~n}izares, J.~W. Simpson-Porco, E.~Nasr, L.~Fan,
  P.~A. Mendoza-Araya, R.~Tonkoski, U.~Tamrakar, N.~Hatziargyriou, D.~Lagos,
  \emph{et~al.}, ``Microgrid stability definitions, analysis, and examples,''
  \emph{IEEE Transactions on Power Systems}, vol.~35, no.~1, pp. 13--29, 2019.

\bibitem{majumder2013some}
R.~Majumder, ``Some aspects of stability in microgrids,'' \emph{IEEE
  Transactions on Power Systems}, vol.~28, no.~3, pp. 3243--3252, 2013.

\bibitem{mendoza2014impedance}
P.~A. Mendoza-Araya and G.~Venkataramanan, ``Impedance matching based stability
  criteria for ac microgrids,'' in \emph{2014 IEEE Energy Conversion Congress
  and Exposition (ECCE)}.\hskip 1em plus 0.5em minus 0.4em\relax IEEE, 2014,
  pp. 1558--1565.

\bibitem{xu2016microgrids}
Y.~Xu, C.-C. Liu, K.~P. Schneider, F.~K. Tuffner, and D.~T. Ton, ``Microgrids
  for service restoration to critical load in a resilient distribution
  system,'' \emph{IEEE Transactions on Smart Grid}, vol.~9, no.~1, pp.
  426--437, 2016.

\bibitem{yan2018small}
Y.~Yan, D.~Shi, D.~Bian, B.~Huang, Z.~Yi, and Z.~Wang, ``Small-signal stability
  analysis and performance evaluation of microgrids under distributed
  control,'' \emph{IEEE Transactions on Smart Grid}, vol.~10, no.~5, pp.
  4848--4858, 2018.

\bibitem{sanders1991generalized}
S.~R. Sanders, J.~M. Noworolski, X.~Z. Liu, and G.~C. Verghese, ``Generalized
  averaging method for power conversion circuits,'' \emph{IEEE Transactions on
  Power Electronics}, vol.~6, no.~2, pp. 251--259, 1991.

\bibitem{guo2014dynamic}
X.~Guo, Z.~Lu, B.~Wang, X.~Sun, L.~Wang, and J.~M. Guerrero, ``Dynamic
  phasors-based modeling and stability analysis of droop-controlled inverters
  for microgrid applications,'' \emph{IEEE Transactions on Smart Grid}, vol.~5,
  no.~6, pp. 2980--2987, 2014.

\bibitem{vorobev2017high}
P.~Vorobev, P.-H. Huang, M.~Al~Hosani, J.~L. Kirtley, and K.~Turitsyn,
  ``High-fidelity model order reduction for microgrids stability assessment,''
  \emph{IEEE Transactions on Power Systems}, vol.~33, no.~1, pp. 874--887,
  2017.

\bibitem{bottrell2013dynamic}
N.~Bottrell, M.~Prodanovic, and T.~C. Green, ``Dynamic stability of a microgrid
  with an active load,'' \emph{IEEE Transactions on Power Electronics},
  vol.~28, no.~11, pp. 5107--5119, 2013.

\bibitem{li2004design}
Y.~Li, D.~M. Vilathgamuwa, and P.~C. Loh, ``Design, analysis, and real-time
  testing of a controller for multibus microgrid system,'' \emph{IEEE
  Transactions on Power Electronics}, vol.~19, no.~5, pp. 1195--1204, 2004.

\bibitem{guerrero2011hierarchical}
J.~M. Guerrero, J.~C. Vasquez, J.~Matas, L.~G. De~Vicu{\~n}a, and M.~Castilla,
  ``Hierarchical control of droop-controlled ac and dc microgrids—a general
  approach toward standardization,'' \emph{IEEE Transactions on Industrial
  Electronics}, vol.~58, no.~1, pp. 158--172, 2011.

\bibitem{du2019comparative}
W.~Du, Z.~Chen, K.~P. Schneider, R.~H. Lasseter, S.~Pushpak, F.~K. Tuffner, and
  S.~Kundu, ``A comparative study of two widely used grid-forming droop
  controls on microgrid small signal stability,'' \emph{IEEE Journal of
  Emerging and Selected Topics in Power Electronics}, 2019.

\bibitem{rasheduzzaman2015reduced}
M.~Rasheduzzaman, J.~A. Mueller, and J.~W. Kimball, ``Reduced-order
  small-signal model of microgrid systems,'' \emph{IEEE Transactions on
  Sustainable Energy}, vol.~6, no.~4, pp. 1292--1305, 2015.

\bibitem{mariani2014model}
V.~Mariani, F.~Vasca, J.~C. V{\'a}squez, and J.~M. Guerrero, ``Model order
  reductions for stability analysis of islanded microgrids with droop
  control,'' \emph{IEEE Transactions on Industrial Electronics}, vol.~62,
  no.~7, pp. 4344--4354, 2014.

\bibitem{slotine1991applied}
J.-J.~E. Slotine, W.~Li, \emph{et~al.}, \emph{Applied nonlinear control}.\hskip
  1em plus 0.5em minus 0.4em\relax Prentice hall Englewood Cliffs, NJ, 1991,
  vol. 199, no.~1.

\bibitem{householder2013theory}
A.~S. Householder, \emph{The theory of matrices in numerical analysis}.\hskip
  1em plus 0.5em minus 0.4em\relax Courier Corporation, 2013.

\bibitem{ieeetestfeeder}
{IEEE~PES~AMPS~DSAS~test~feeder~working~group~and~others}, ``Test feeders,''
  \url{http://sites.ieee.org/pes-testfeeders/}, 2000.

\bibitem{kersting2006distribution}
W.~H. Kersting, \emph{Distribution system modeling and analysis}.\hskip 1em
  plus 0.5em minus 0.4em\relax CRC press, 2006.

\bibitem{pscad}
{H.~V.~D.~C~Research~Center}, ``{PSCAD},'' \url{https://hvdc.ca/pscad/}, 2019.

\end{thebibliography}
\end{document}